\documentclass[11pt]{article}

\usepackage[a4paper, total={6in, 8in}]{geometry}
\usepackage[inline]{enumitem}
\usepackage{amsmath,amsfonts,amssymb,amsthm,appendix,array,bbm,enumitem,natbib,epsfig,epstopdf,stackrel,setspace,titling,url,abstract,mathtools}
\usepackage[]{hyperref} 
\title{Asymptotic bayes optimality under sparsity for equicorrelated multivariate normal test statistics}
\author{Rahul Roy \and Subir Kumar Bhandari}
\date{}
\begin{document}
	\maketitle
	\theoremstyle{plain}
	\newcounter{dummy} \numberwithin{dummy}{section}
	\newtheorem{thm}[dummy]{Theorem} 
	\newtheorem{prpsn}{Proposition} 
	\newtheorem{lma}[dummy]{Lemma}
	\newtheorem{crlry}{Corollary}
	\theoremstyle{definition}
	\newtheorem{defn}[dummy]{Definition} 
	\newtheorem{asmtn}{Assumption}
	\begin{abstract}
		Here we address dependence among the test statistics in connection with asymptotically Bayes' optimal tests in presence of sparse alternatives. Extending the setup in \cite{bogdanetal11} we consider an equicorrelated ( with equal correlation $\rho$ ) multivariate normal assumption on the joint distribution of the test statistics, while conditioned on the mean vector $\boldsymbol{\mu}$. Rest of the set up is identical to \cite{bogdanetal11} with a slight modification in the asymptotic framework. We exploit an well known result on equicorrelated multivariate normal variables with equal marginal variances to decompose the test statistics into independent random variables. We then identify a set of independent yet unobservable gaussian random varibales sufficient for the multiple testing problem and chalk out the necessary and sufficient conditions for single cutoff tests to be ABOS based on those dummy variables following \cite{bogdanetal11}. Further we replaced the dummy variables with deviations of the statistics from their arithmatic means which were easily calculable from the observations due to the decomposition used earlier. Additional assumptions are then derived so that the necessary and sufficient conditions for single cutoff tests to be ABOS using the independent dummy variables plays the same role with the replacement variable as well (with a deviation of order $o(1)$).  Next with the same additional assumption, necessary and sufficient conditions for single cutoff tests to control the Bayesian FDRs are derived and as a consequence under various sparsity assumptions we proved that the classical Bonferroni and Benjamini-Hochberg methods of multiple testing are ABOS if the same conditions are satisfied.

	\end{abstract}
	
	\section{Introduction:}
	With the surge of technological innovations towards the end of the last century, accumulation of large data was emminent. Devising methods to appropriately combine decisions for several hypotheses was therefore a necessity. Fields such as genome wide association study (GWAS), epidemiology, astronomy, stock market analysis, meteorology require the applications of multiple hypotheses testing. Throughout the years statisticians have followed two vastly different approaches for testing multiple hypotheses, namely the frequentist or classical approach and the Bayesian approach. The main differences between the two approaches lies in how the problem of multiplicity adjustment is dealt with (\cite{Sjlander2019}). In the frequentist paradigm, methods are so designed that the adjustment of the significance level of each test leads to the control of some false positive error metrices like $FWER$ (\cite{sidak67, Holm79, HC88, Hoch88, Hom88, Rom90, DuSV04, NiYKS21}), $FDR$ (\cite{BH95, BY01, BY05, BKY06, GL02, Sar02, Sar04, Sar06, Sar08a, Storey02, STS04, FDR07, FDR08, GBS09}) or $k-FWER$ ( \cite{LR2005, DVLP04, VLDP04}). And a multiple testing method is considered optimal if some measure of true positive is maximized among all the methods for which the measure of false positives are controlled at some fixed level. Meanwhile Bayesian methods tend to minimize the Bayes risk  based on certain loss function. Often Bayesian optimal rules are optimal in a classical sense too (\cite{SC07, SC09, CS09, Sea15}).  
	
	With extensive simulation study in \cite{Bogdanetal07, bogdan08} showed that performance of the \cite{BH95} and some empirical bayes methods are equivalent under sparsity when the test statistics follow independent gaussian distribution. Later \cite{bogdanetal11, FB13} showed theoretically that performance of the \cite{BH95} is asymptotically equivalent to the bayes oracle under different sparsity conditions. This too was also proved under an independent gaussian assumptoion. Several extension of \cite{bogdanetal11} has been studied in the literature. The Subbotin family of distribution for test statistics was shown to yield similar results (\cite{NR12}). A full Bayes' treatment was adopted for 0-1 risk in \cite{DG13} with assumption of a Horseshoe prior on the parameters.
	
	In this paper we have considered an equicorrelated multivariate normal distribution for the test statistics. In recent works on multiple testing for similar distributional assumption, \cite{DB21} provided a large scale upper bound on the $FWER$ value obtained for bonferroni corrections; later \cite{dey21} derived an upper bound for the same but in a non-asymptotic set up. 
	
	 We have tried to extend the results in \cite{bogdanetal11} for the equicorrelated set up with a little adjustment in the assumptions. The article is organised in the following way: in Section \ref{section 2}, the underlying set up is described. Lemma \ref{lemma1} is also mentioned here which enables us to use the theorems in \cite{bogdanetal11} directly. And finally we have stated the Assuption \ref{asmtn1} which is a slightly changed version of Assumption A in \cite{bogdanetal11}.  In Section \ref{section 3} we stated and proved the main result describing the conditions for single cutoff tests to be ABOS. In section \ref{section 4} the conditions for single cutoff tests to asymptotically control BFDR is described. We finally showed that the Benjamini Hochberg procedure and Bonferonni adjustment of multiple testing are ABOS in certain sparsity set ups as long as the conditions for controlling BFDR are satisfied. We conclude with some discussions in Section \ref{section 5}.
	 
	\section{Setup and assumptions:} \label{section 2}
	Suppose the observations are represented as an $m$ dimensional vector $\mathbf{X'}=(X_1,X_2,\cdots,$ $X_m)$. Let $\mathbf{X}\sim MVN(\boldsymbol{\mu},\Sigma_\epsilon)$. Here $\boldsymbol{\mu'}=(\mu_1,\mu_2,\cdots,\mu_m)$ and 
	$(\Sigma_\epsilon)_{ij} = \sigma_\epsilon^2 \rho ( \sigma_\epsilon^2 ) $ if $i\neq j ( i= j) $. $\mu_i (\in\mathbb{R})$ represents effect under investigation; $\sigma_\epsilon^2(>0)$ is the variance of the random noise and $\rho(\in[0,1])$ is the equal correlation between each pair of  coordinates. We assume that marginally $\mu_i \sim N(0,\sigma_i^2) \ i=1(1)m$  and are independent. Their distribution depends on unobservable random vector $\boldsymbol{\theta'}=(\theta_1,\theta_2,\cdots,\theta_m)$ with $\theta_i\stackrel{iid}{\sim} \text{Bernoulli}(p)$ such that: \begin{equation*}
		\sigma_i^2=\begin{cases}
		\sigma_0^2 & \text{ if } \theta_i = 0\\
		\sigma_0^2 + \tau^2 & \text{ if } \theta_i = 1
	\end{cases} \  \forall 1(1)m.
	\end{equation*} Here $\sigma^2_0\geq0$ represents the small null effects exerted by each coordinates and $\tau^2>0$ is representative of the high impact of significant coordinates. \cite{bogdanetal11} presents a great account on the rationality of these assumtions. Our objective here is to find out which coordinates has significant impact on the observations. i.e., we aim to test the following hypotheses \begin{equation*}
	H_{0i}: \mu_i \sim N(0,\sigma_0^2) \text{ vs. }  H_{Ai}: \mu_i \sim N(0,\sigma_0^2+\tau^2) \ \forall i = 1(1)m.
\end{equation*}
	
	Marginally, 
	\begin{equation*}
		\mu_i \stackrel{iid}{\sim} p N(0,\sigma_0^2+\tau^2) + (1-p) N(0,\sigma_0^2) \ \forall i = 1(1)m.
	\end{equation*}
	
	
	The following lemma decomposes equicorrelated normal variables into independent normal variables.
	
	\begin{lma}
		\label{lemma1}
		Let $\mathbf{U'} = (U_1, U_2, \cdots , U_n) \sim MVN (0,(1 -r)I + r \mathbf{11}' ).$ Then,
		\[
		\exists
		\begin{rcases}
			Z_i &\stackrel{iid}{\sim}  N(0,(1-r)) \forall i=1(1)m\\
			V & \sim N(0, r)
		\end{rcases}\text{independent}
		\]
		Such that, $(U_1, U_2, \cdots , U_n) \stackrel{d}{=} (Z_1 + V, Z_2 + V,\cdots , Z_n + V ) \forall i = 1(1)n$
	\end{lma}
	Proof of Lemma \ref{lemma1} is straightforward and therefore omitted.
	By Lemma \ref{lemma1}, we can write, given $\boldsymbol{\mu}$ and $\boldsymbol{\theta}$, 
	
	\begin{equation}
		\label{decomp}
		X_i = Z_i + Q, \ \forall i=1(1)m
	\end{equation}
	where, 
	\[
	\begin{rcases}
		Z_i & \stackrel{indep}{\sim} N(\mu_i,\sigma^2_\epsilon(1-\rho))\\
		Q & \sim N(0, \sigma^2_\epsilon\rho)
	\end{rcases} \text{independent}
	\]
	
	Note that, the observation $\mathbf{X}$ depends on the parameter in question i.e., $\boldsymbol{\theta}$ only through $\boldsymbol{\mu}$. Now due to \ref{decomp}, we have, $\mathbf{X} = \mathbf{Z} + Q \mathbf{1}$, where $\mathbf{Z'} = (Z_1,Z_2, \cdots, Z_m)$. Since $Q$ and $\mathbf{Z}$ are independent and distribution of $\mathbf{Z}$ involves $\boldsymbol{\mu}$ but that of $Q$ is free from $\boldsymbol{\mu}$, it is evident that $\mathbf{Z}$ is a sufficient statistic and $Q$ is an ancilliary statistic for parameter $\boldsymbol{\mu}$ and therefore $\boldsymbol{\theta}$. Therefore, inference on  $\boldsymbol{\theta}$ based on $\mathbf{Z}$ will convey the same information as inference based on $\mathbf{X}$ would have. Considering $\mathbf{Z}$ as our observations, our problem can be rewrittrn as follows:
	
	Suppose,$\forall i =1(1)m$, \begin{align} 
		Z_i \stackrel{indep}{\sim} & N(\mu_i,\sigma_{\epsilon,\rho}^2) \label{eqnstrt}\\
		\mu_i \stackrel{indep}{\sim} & \begin{cases}
			N(0,\sigma_0^2) & \text{if } \theta_i = 0\\
			N(0,\sigma_0^2+\tau^2) & \text{if } \theta_i =1
		\end{cases}\\
		\theta_i \stackrel{iid}{\sim} &\text{ Bernoulli}(p)\label{eqnend}
	\end{align}
	Where $\sigma_{\epsilon,\rho}^2=\sigma^2_\epsilon(1-\rho)$. Our aim is to test the series of $m$ hypotheses $H_{0i}: \theta_i =0$ vs. $H_{Ai} : \theta_i =1$ $\forall i = 1(1)m$. The problem described in \ref{eqnstrt} to \ref{eqnend} along with the $m$ pairs of hypotheses to be tested is the same problem considered in \cite{bogdanetal11}. The only difference is that, the observations $\mathbf{Z}$ considered here is actually an unobservable theoretical random vector. However, if we make appropriate modifications in the assumptions considered in \cite{bogdanetal11}, the results obtained there will go through here too.
	
	Suppose $\delta_0$ ($\delta_A$) be the common loss incurred for each coordinates ($i$) when $H_{0i}$ ($H_{Ai}$) is falsely rejected. Then the Bayes rule for $i$-th coordinate will reject $H_{0i}$ when:
	
	\begin{equation}
		\label{bayesrle}
		\frac{\phi_A(Z_i)}{\phi_0(Z_i)} \geq \frac{(1-p)\delta_0}{p\delta_A}
	\end{equation}
	
	where $\phi_A$ and $\phi_0$ are respectively the alternative and null distributions of $Z_i$. The optimal rule in \ref{bayesrle} can be simplified as:
	\begin{equation}
		\text{Reject } H_{0i} \qquad \text{if } \frac{Z_i^2}{\sigma^2}\geq c^2
	\end{equation}
	
	where $\sigma^2=\sigma_{\epsilon,\rho}^2+\sigma_0^2$ and 
	\begin{equation}
		c^2 = c_{\tau,\sigma,p,\delta_A,\delta_0}^2 = \frac{\sigma^2+\tau^2}{\tau^2}\big( \log\big(\big(\frac{\tau}{\sigma}\big)^2+1\big)+2\log\big(\frac{(1-p)\delta_0}{p\delta_A}\big)\big)
	\end{equation}
	For the sake of simpliciy and convenience, let, $ u = \big(\frac{\tau}{\sigma}\big)^2 $ and $ v = u \big(\frac{(1-p)\delta_0}{p\delta_A}\big)^2$. Then the cutoff $c^2$ can be rewritten as:
	\begin{equation}
		c^2= c_{u,v}^2= \big(1+\frac{1}{u}\big)(\log v + \log(1+\frac{1}{u}))
	\end{equation}
	
	
	We assume that the equal correlation ($\rho$) between different coordinate pairs of the original observation vector remain constant. Let $\boldsymbol{\gamma_t'}=(\delta_{0t},\delta_{At}, p_t, \sigma_{\epsilon t},\sigma_{0t})$ be a sequence of parameter vectors. Asymptotic optimality of the multiple testing shall be examined as $t \uparrow \infty$.  In particular, one interesting case is when the proportion of alternatives $p_t\rightarrow 0$ as $t\uparrow \infty$. The discussion in section 2.1 in \cite{bogdanetal11}, justifies some additional assumptions on the asymptotic framework too.  
	\begin{asmtn} \label{asmtn1}
		The sequence of parameter vectors $\boldsymbol{\gamma_t}$ is said to follow the asymptotic framework of ABOS if,
		\begin{enumerate*}
			\item $ p_t \rightarrow 0$,\item $ u_t \rightarrow \infty$,\item $ v_t \rightarrow \infty$ \&
			\item $\frac{\log v_t}{u_t} \rightarrow C\in (0,\infty)$.
		\end{enumerate*} 
		
		Here, the change in the parameters with the number of coordinates $m$ shall be studied. i.e. the subscript $t$ is replaced with $m$. 2 different sparse case shall be considered; the first being extremely sparse case as chracterized by 
		
		\begin{equation} 
			\label{extrmsprs}
			mp_m \rightarrow s\in(0,\infty] \quad \text{and } \quad \frac{\log(mp_m)}{\log(m)}\rightarrow 0
		\end{equation} 
		
		For example, the case $p_m\propto \frac{1}{m}$ satisfies conditions of extreme sparse alternative. The second ``denser'' case is characterized by:
		\begin{equation} 
			\label{mdmsprs}
			p_m \rightarrow 0 \quad \text{and } \quad \frac{\log(mp_m)}{\log(m)}\rightarrow C_p \in (0,1]
		\end{equation} 
	\end{asmtn} 
	\section{Asymptotic Bayes Optimality under Sparsity(ABOS):} \label{section 3}
	According to Theorem 3.2 in \cite{bogdanetal11}, given assumption \ref{asmtn1} is true, any rule that rejects $H_{0i}$ if $Z_i^2\geq c_m \ \forall i = 1(1)m$ where, $c_m =\log v_m + l_m$, with, 
	\begin{align}
		\label{lm1} l_m =  o(\log v_m)\\
		\label{lm2} l_m + 2\log \log v_m \rightarrow \infty 
	\end{align}
	will be an asymptotically bayes optimal test under sparsity.
	
	But, as mentioned earlier, we cannot directly observe $\mathbf{Z}$. However, averaging \ref{decomp} for all the $m$ coordinates,
	
	\begin{equation}
		\label{decompbar}
		\bar{X}= Q + \bar{Z} 
	\end{equation}
	Subtracting \ref{decompbar} from \ref{decomp} we get,
	\begin{equation}
		X_i-\bar{X}= Z_i-\bar{Z} \ \forall i = 1(1)m
	\end{equation}
	Let, $U_i = Z_i -\bar{Z} \ \forall i = 1(1)m.$ It is evident from \ref{decompbar} that $\mathbf{U'}=(U_1,U_2,\cdots,U_m)$ can be easily obtained from $\mathbf{X}$ as, $ \mathbf{U}=\mathbf{X} - (\frac{1}{n}\mathbf{X'1}) \mathbf{1}$.
	
	\begin{thm}
		\label{thm1}
		Let assumption \ref{asmtn1} holds. Consider any testing rule of the form: \begin{equation*}
			\forall i=1(1)m \text{ reject } H_{0i} \text{ if } (\frac{U_i}{\sigma})^2 \geq c_m
		\end{equation*}
		then if $\log (\frac{\delta_m}{p_m}) =  o_m( m)$, such a rule will be $ABOS$- iff,  $c_m = \log v_m + l_m $ and $l_m$ follows \ref{lm1} and \ref{lm2}. 
	\end{thm}
	\begin{proof}
		Note that, $(\frac{U_i}{\sigma})^2\geq c_m$ iff $\frac{Z_i^2}{\sigma^2}\geq \min((\sqrt{c_m}-\frac{\bar{Z}}{\sigma})^2,(\sqrt{c_m}+\frac{\bar{Z}}{\sigma})^2)$. i.e., $\frac{Z_i^2}{\sigma^2}\geq c_m+\frac{\bar{Z}^2}{\sigma^2}-2\frac{\sqrt{c_m}}{\sigma}|\bar{Z}|$. Now, since, $\frac{Z_i}{\sigma}\stackrel{iid}{\sim}N(0,p_m(1+u_m)+(1-p_m))$, by SLLN, $\bar{Z}\stackrel{a.s}{\rightarrow}0$.
		\begin{lma}
			Suppose, assumption \ref{asmtn1} is satisfied  and $c_m$ is defined as in \ref{thm1} Then,
			$\sqrt{c_m}\frac{|\bar{Z}|}{\sigma}=o_m(1)$ w.p.1. if, $\log (\frac{\delta_m}{p_m}) = o_m( m)$.
		\end{lma}
		\begin{proof}
			Assumption \ref{asmtn1} implies that,
			\begin{equation}
				\label{asm2}
				\frac{\log v_m}{u_m} \rightarrow C \in (0,\infty).
			\end{equation}
			Now, $v_m=u_m f_m^2 \delta_m ^2$, where, $f_m = \frac{1-p_m}{p_m} $ and $\delta_m = \frac{\delta_0}{\delta_A}$.
			So by simple calculations and due to the fact that $p_m \rightarrow 0$, we get 
			\begin{equation}
				\log \delta_m - \log p_m -\frac{C}{2} u_m \rightarrow 0
			\end{equation}
			i.e., 
			\begin{equation}
				\label{logvpf}
				2\log (\frac{\delta_m}{p_m}) -\log v_m \rightarrow 0
			\end{equation}
			
			Now, $\frac{Z_i}{\sigma}\stackrel{iid}{\sim} N(0,\sigma_m^2)$ where, $\sigma_m^2 = p_m(1+u_m)+(1-p_m)$. It is easy to see that $\sigma_m^2\rightarrow 1$. So, $\frac{\sqrt{m}\bar{Z}}{\sigma_m\sigma}\sim N(0,1)$. Therefore, $\frac{a_m\bar{Z}}{\sigma}\rightarrow 0 $ a.s. where $a_m=o_m(\sqrt{m})$.
			
			By definition,  $c_m = \log v_m + l_m $ where $l_m$ follows \ref{lm1} and \ref{lm2}. So if  $ \log (\frac{\delta_m}{p_m}) = o_m( m)$, then due to \ref{logvpf},
			
			\begin{equation}
				\sqrt{c_m}=o_m(\sqrt{m})
			\end{equation}	
			which proves the lemma.
		\end{proof}
		
		Now, by SLLN, $(\frac{\bar{Z}}{\sigma})^2 \rightarrow 0$ a.s. i.e., $(\frac{\bar{Z}}{\sigma})^2 = o_m(1).$ a.s. So the final cutoff is,\begin{equation}
			\begin{split}
				c_{1m}=& c_m+\frac{\bar{Z}^2}{\sigma^2}-2\frac{\sqrt{c_m}}{\sigma}|\bar{Z}|\\
				= & \log v_m + l_m + 
				\frac{\bar{Z}^2}{\sigma^2}-2\frac{\sqrt{c_m}}{\sigma}|\bar{Z}|\\
				= & 
				\log v_m + l_{1m} 
			\end{split}
		\end{equation}
		where, $l_{1m}$ almost surely follows the required assumptions in Theorem 3.2 in \cite{bogdanetal11}. Finally, since an $ABOS$ rule requires ratio of two risks (i.e. expectations) to converge to 1, if for a rule the conditions are satisfied on an event with probability 1 (or on a sequence of events, whose probability converges to 1) the rule will be $ABOS$. Hence the Theorem is proved.
	\end{proof}
	\section{ABOS for Bayesian FDR control:} \label{section 4} As defined by \textbf{Efron and Tibshirani(2002)} Bayesian FDR ($BFDR$) is defined as: 
	\begin{equation}
		\label{bfdr}
		BFDR = \mathbb{P}(H_{0i} \text{ is true }|H_{0i} \text{ is rejected })=\frac{(1-p)t_1}{(1-p)t_1+p(1-t_2)}
	\end{equation}
	Where as mentioned before, $t_1$ and $t_2$ are respectively the probabilities of type $I$ and type $II$ errors. Now, due to the MLR property and N-P lemma, rejection of $H_{0i}$ is justified for large values of $\frac{Z_i^2}{\sigma^2}.$ Since $U_i^2=Z^2_i + O(\bar{Z})$, and $\bar{Z}\rightarrow 0$ almost surely, for large $m$, a high value of $\frac{U_i^2}{\sigma^2}$ will almost surely indicate the falsehood of $H_{0i}$ and therefore we shall reject $H_{0i}$ if $\frac{U_i^2}{\sigma^2}$ exceeds some fixed constant $c_B^2$. According to Therorem \ref{thm1}, for the test to be ABOS, $c_B$ must follow the properties of the cutoff value mentioned in theorem \ref{thm1}. i.e., we require $\frac{c_B}{\sqrt{u_m+1}}\rightarrow \sqrt{C}$. $C$ being the constant mentioned in Assuption \ref{asmtn1}. Expressions of $t_1$ and $t_2$ for the rule: ``Reject $H_{0i}$ when $\frac{Z_i^2}{\sigma^2}\geq C_m$" can be obtained as:
	\begin{align}
		t_1 = & 2(1-\Phi(\sqrt{C_m}))\\
		t_2 = & 2\Phi(\frac{\sqrt{C_m}}{\sqrt{u_m+1}})-1
	\end{align} 
	And since $U_i^2=Z^2_i + O(\bar{Z})$, and $\bar{Z}\rightarrow 0$ almost surely, the rule  ``Reject $H_{0i}$ when $\frac{U_i^2}{\sigma^2}\geq c_B^2$" will have $t_1$ and $t_2$ as 
	\begin{align}
		t_1 = & 2(1-\Phi(c_B))(1+o_m(1))\\
		t_2 = & (2\Phi(\frac{c_B}{\sqrt{u_m+1}})-1)(1+o_m(1))
	\end{align}almost surely. 
	Equationg the expression \ref{bfdr} of $BFDR$ to $\alpha=\alpha_m$, we get ,
	
	\begin{equation}
		\label{bfdrabos1}
		\frac{(1-\Phi(c_B))}{1-\Phi(\frac{c_B}{\sqrt{u_m+1}})} = \frac{r_{\alpha_m}}{f}
	\end{equation}
	where, $r_\alpha=\frac{\alpha}{(1-\alpha)}$.
	
	From the above discussion, it is evident that an additional assumption to $c_B$ following assumptons in theorem \ref{thm1} would be: 
	\begin{equation}
		\frac{r_{\alpha_m}}{f} \rightarrow 0
	\end{equation}
	
	The following theorem sums up the requirement of a fixed threshold rule of the form ``Reject $H_{0i}$ when $\frac{U_i^2}{\sigma^2}\geq c_B^2$" with $BFDR \ \alpha=\alpha_m$ to be ABOS almost surely. 
	
	\begin{thm}
		\label{thm2}
		Suppose assumption \ref{asmtn1} is satisfied. Consider any fixed threshold rule with $BFDR \ \alpha = \alpha_m \in (0,1)$ with cutoff $c_B^2$. Let, 
		
		\begin{equation}
			\frac{\log(f_m\delta_m\sqrt {u_m})}{\log(f_m/r_{\alpha_m})} = 1+s_m
		\end{equation}
		Then  if, 
		\begin{equation}
			\label{thm23}
			\log (\frac{\delta_m}{p_m}) = o_m(m)
		\end{equation}
		the rule is $ABOS$ iff 
		\begin{equation}
			\label{thm21}
			s_m \rightarrow 0
		\end{equation}
		
		\begin{equation}
			\label{thm22}
			2s_m\log(\frac{f_m}{r_{\alpha_m}})-\log \log (\frac{f_m}{r_{\alpha_m}})\rightarrow -\infty
		\end{equation} 
		and finally 
		
	\end{thm}
	\begin{proof}
		As discussed earlier, $\frac{U_i^2}{\sigma^2}\geq c_B^2$ iff $\frac{Z_i^2}{\sigma^2}\geq c_B^2 + l_{2m}$. Where, $l_{2m} = o_m(1)$ almost surely, if \ref{thm23} holds. 
		
		Theorem 4.1 in \cite{bogdanetal11} shows that, for a fixed threshold rule that rejects $H_{0i}$ if $\frac{Z_i^2}{\sigma^2}\geq c_m^2$ with $BFDR=\alpha_m$, we have
		
		\begin{equation}
			c_m^2 = 2 \log(\frac{f_m}{r_{\alpha_m}})-\log(2\log(\frac{f_m}{r_{\alpha_m}}))+C_1+o_m
		\end{equation} with $C_1=\log(\frac{2}{\pi D^2})$ and $ D= 2(1-\Phi(\sqrt C))$.	 
		From section 10 in supplimentary to \cite{bogdanetal11} it has been proved that $c_m^2 = \log v_m +l_m$ where $l_m$ follows \ref{lm1} and \ref{lm2} if and only if \ref{thm21} and \ref{thm22} occurs. So, the rest of the proof follows form Theorem \ref{thm1}. 
	\end{proof}
	\subsection{Optimality of the asymptotic approximation to the BH threshold:}  \textbf{Genovese and Wasserman(2002)} proves that if $m \uparrow \infty$, and $p_m$ remains constant, then the threshold of BH can be approximated by
	\begin{equation}
		\label{approxbh}
		c_{GW}: \frac{1-\Phi(c_{GW})}{(1-p)(1-\Phi(c_{GW}))+p(1-\Phi(c_{GW}/\sqrt{u+1})}=\alpha
	\end{equation}
	It is evident that, the expression in \ref{approxbh} lacks the multiple $(1-p)$ in the numerator in the expression of $BFDR$ in \ref{bfdr}. So if $p_m \rightarrow 0$, the cutoff $c_B$ of $BFDR$ control must converge to $c_GW$. The following theorem formalises the claim.
	
	\begin{thm}
		Consider the fixed cutoff rule rejecting $H_{0i}$ if $\frac{U_i^2}{\sigma^2}\geq c_{GW}^2$. If \ref{thm23} holds, then, this rule is $ABOS$ iff the $BFDR$ rule defined in \ref{bfdr} is $ABOS$ and in this case, 
		\begin{equation}
			c_{GW}^2 = c_B^2 + o_m(1)
		\end{equation}
		almost surley.
	\end{thm}
	\begin{proof}
		The proof follows from the proof of Theorem 4.2 in \cite{bogdanetal11}. and Lemma \ref{lemma1}.
	\end{proof}
	
	\subsection{ $ABOS$ of the Bonferroni correction:}
	
	Here, $H_{0i}$ is rejected if $\frac{|U_i|}{\sigma}\geq c_{Bon}$
	where, \begin{equation}
		c_{Bon}: 1-\Phi(c_{Bon}) =\frac{\alpha}{2m}
	\end{equation}
	which, for large $m$, can be equivalently written as, 
	\begin{equation}
		c_{Bon}^2 = 2\log(\frac{m}{\alpha})-\log(2(\log(\frac{m}{\alpha})))+\log(2/\pi)+o_m
	\end{equation}
	\begin{lma}
		Suppose \ref{extrmsprs} holds. The Bonferonni procedure at $FWER$ level $\alpha_m\rightarrow \alpha_\infty\in[0,1)$ is $ABOS$ if the assumptions of Theorem \ref{thm2} are satisfied.
	\end{lma}
	\begin{proof}
		This lemma is a direct concequence of Theorem \ref{thm2} and Lemma 5.1 in \cite{bogdanetal11}.
	\end{proof}
	\subsection{$ABOS$ of $BH$:} According to the $BH$ rule, if $p_{(1)}\leq p_{(2)}\leq \cdots\leq p_{(m)}$ be the ordered p-valus of the $m$ hypotheses, then to control $FDR$ at level $\alpha,$ we need to reject those hypotheses with p value less than $p_{(k)}$ where, 
	
	\begin{equation}
		k = \max \bigg\{ i: p_{(i)} \leq \frac{i\alpha}{m}\bigg \}
	\end{equation}
	
	Let $(1-\hat{F}_m(y))  = \#\{|Z_i|\geq y\}/m$. Then $BH$ procedure rejects $H_{0i}$ when $Z_i^2 \leq \tilde{c}_{BH}^2$ where
	\begin{equation}
		\tilde{c}_{BH}= \inf\bigg\{ y: \frac{2(1-\Phi(y))}{(1-\hat{F}_m(y))}\leq \alpha\bigg\}
	\end{equation}
	Also, $BH$ rejects $H_{0i}$ whenever $Z_i^2 \leq c_{Bon}^2$.
	So, the random threshold for $BH$  can be written as 
	\begin{equation}
		c_{BH}= \min (c_{Bon},\tilde{c}_{BH})
	\end{equation}
	
	\begin{thm}
		Suppose, $\{\alpha_m\}$ is a sequence in $(0,1)$ such that $\alpha_m\rightarrow \alpha_\infty<1$. 
		Then if 
		\begin{equation*}
			p_m \rightarrow 0 , \quad mp_m\rightarrow s\in(0,\infty] \quad \text{as} \quad m\uparrow \infty,
		\end{equation*} then $BH$ at $FDR$ level $\alpha=\alpha_m$ is $ABOS$ if \ref{thm23}, \ref{thm21} and \ref{thm22} hold.
	\end{thm}
	
	\begin{proof}
		Following Theorem 5.2. in  \cite{bogdanetal11} and  lemma \ref{lemma1} the proof of the theorem is straightforward.
	\end{proof}
	\section{Discussion:} \label{section 5}
	We have considered $\rho$ to be non-negative mainly because we are studying the asymptotic proprties as $m\uparrow\infty$ but we are considering $\rho$ to be fixed. Therefore for the correlation matrix to be non negative definite we require $\rho \geq \lim_{m\uparrow\infty} -1/m$. i.e., $\rho \geq 0$. Note that, Assumption \ref{asmtn1} is identical to  Assumption A in \cite{bogdanetal11} if we consider $\rho$ to be fixed and known. However our analysis is still viable if $\rho_m$ changes with $m$ following Assumption \ref{asmtn1} given the following conditions hold:
	
	\begin{equation*}
		\begin{split}
			\rho_m \in [-\frac{1}{m},1] \\
			\lim_{m\uparrow\infty} \rho_m = \rho \in [0,1]
		\end{split}
	\end{equation*} 
	
	$FDR$ controlling procedures under dependence has been a key subject of research for last two decades. A study of the existing literature is available in \cite{Ghosh19}. In this article we proved the asymptotic optimality of the Benjamini Hochberg rule for 2 different sparsity levels (\ref{extrmsprs},\ref{mdmsprs}) for the equicorrelated assumption. More general dependence adaptations of the same problem is a possible extension.  
	
	Considering non gaussian assumption on the joint distribution. An immidiate follow up can be to consider the Subbotin class of distributions as adopted in \cite{NR12}. However since a result of the form Lemma \ref{lemma1} is absent for non gaussian distributions, this adaptation maight be tricky and may require more invasive mathematcal results.
	
	We have considered the parameters $\sigma$ and $\rho$ to be known and followed an empirical Bayes' path for the analysis. Instead we could have used a fulll Bayesian treatment considering the parameters to be unknown and assuming appropriate prior distributions. A Horseshoe prior as considered in \cite{DG13} might be appropriate for this study.

	

\end{document}